\documentclass[aps,prd,final,twocolumn,twoside,letterpaper]{revtex4}
\usepackage{amsmath}
\usepackage{amssymb}
\usepackage{amsthm}
\usepackage{enumerate}
\usepackage{array}
\usepackage{breqn}
\usepackage[margin=1in]{geometry}
\usepackage{tabularx}

\newcommand{\ket}[1]{\left|#1\right\rangle}

\theoremstyle{remark}

\theoremstyle{plain}
\newtheorem{theorem}[subsection]{Theorem}

\theoremstyle{definition}

\begin{document}
\title{Simultaneous Measurement and Entanglement}
\author{Andrey Boris Khesin}
\author{Peter Shor}
\affiliation{Department of Mathematics$,$ Massachusetts Institute of Technology$,$ Cambridge$,$ Massachusetts 02139$,$ USA}
\date{\today}
\begin{abstract}
\fontsize{9pt}{9pt}\selectfont
We study scenarios which arise when two spatially-separated observers, Alice and Bob, are try to identify a quantum state sampled from several possibilities.
In particular, we examine their strategies for maximizing both the probability of guessing their state correctly as well as their information gain about it.
It is known that there are scenarios where allowing Alice and Bob to use LOCC offers an improvement over the case where they must make their measurements simultaneously.
Similarly, Alice and Bob can sometimes improve their outcomes if they have access to a Bell pair.
We show how LOCC allows Alice and Bob to distinguish between two product states optimally and find that a LOCC is almost always more helpful than a Bell pair for distinguishing product states.
\end{abstract}
\fontsize{12pt}{12pt}\selectfont

\maketitle

\pagestyle{myheadings}
\markboth{Khesin and Shor}{Simultaneous Measurement and Entanglement}
\thispagestyle{empty}
\section{Introduction}
In this paper, we consider situations where two parties, whom we will call Alice and Bob, hold correlated, but unentangled, quantum states. Their goal is to perform measurements to learn about the states. The first paper to consider such a situation was the seminal paper of Peres and Wootters \cite{pw}. They considered two experimenters in separate labs who are each given a qubit; the two qubits are in the same state, which is chosen from a set of three pure states. Their work began with the question of whether separate measurements could extract as much information about the state as a joint measurement could. Although they gave strong evidence that this was true, they did not provide a proof. Proofs were eventually provided by Bennet et al.\ \cite{bdfetal} (for a different set of correlated states) and by Chitamber and Hsieh \cite{ch} (for the Peres-Wootters states). 

Peres and Wootters further showed that if separate measurements are required, the optimal measurement needs to be an LOCC (local operations and classical communication) measurement. That is, making simultaneous measurements in both labs is not adequate; a measurement is made in one lab, the outcome is transmitted to the other lab, and the second lab's measurement depends
on the outcome of the first one. They further showed the remarkable fact that the best way of extracting information about their state requires not just two LOCC measurements, but a long sequence of LOCC weak measurements, interspersed by classical communication between the experimenters.

We will be looking at another variation on this theme. In our paper, we will require that the two parties make the measurements simultaneously. That is, the two experiments must make independent measurements in each lab, without communicating with each other during the measurement process. After the measurements, the two experimenters share their information to learn as much as they can about the original states. We will compare the case where the experiments are or are not allowed to share an entangled state before making their measurements. We further compare these scenarios with how much they can learn using LOCC measurements. 

We consider both the case where Alice and Bob are tasked with maximizing the probability of correctly guessing the state that they were given as well as the case where their task is to maximize the information they gain about the identity of their state.
We examine how much better they do at these tasks under various conditions.

Note that allowing Alice and Bob to use both a Bell pair and LOCC would allow Alice to use quantum teleportation to send her qubit to Bob, who would then be able to make joint measurements on the two qubits, removing the challenge imposed by the labs being separated.
Furthermore, we note that if we allowed the possibilities from which we choose the state to be entangled, then we could simply choose a state from the Bell basis.
If Alice and Bob have access to a Bell pair, they can determine the state with certainty by making Bell basis measurements on their half of the state and of the Bell pair.
They would not be able to achieve more than a 0.5 chance of guessing correctly with only LOCC.
For this reason, we limit our investigation to product states.

In this paper, we make a number of claims regarding optimal strategies for Alice and Bob under various conditions for optimizing certain metrics.
These strategies were obtained with the help of an algebra engine in several ways.
The first was an optimization over all the parameters of the measurement with a number of different guesses at good initial conditions.
The second strategy was to fix the initial conditions (such as measuring in the computational basis) and alternating between optimizing Alice's measurement while keeping Bob's fixed, and optimizing Bob's while fixing Alice's.
The fact that these resulted in agreement gives us a lot of confidence about the following results.

\section{Simultaneous Measurement}

Suppose that Alice and Bob are in different laboratories and can only communicate classically.
They are each handed a single qubit and told that their joint state was randomly chosen from the following four states.
\begin{align*}
&\ket{0}_A\ket{0}_B\\
&\ket{0}_A\ket{1}_B\\
&\ket{1}_A\ket{+}_B\\
&\ket{1}_A\ket{-}_B
\end{align*}

Alice and Bob can guarantee that they can identify which of the four states they are holding using only LOCC.
Namely, Alice measures her qubit in the $Z$ basis, identifying it precisely, and then tells Bob the result.
If Alice was holding $\ket0$, Bob measures his qubit in the $Z$ basis and if Alice had $\ket1$, Bob uses the $X$ basis.
This allows Bob to identify his qubit as well.

However, if we now force Alice and Bob to make their measurements simultaneously, we see that although Alice can always determine her state, Bob has to try to guess his state and can succeed with probability $\frac{2+\sqrt2}{4}\approx0.853$ by measuring in the $\frac{X+Z}{\sqrt2}$ basis. This results in an information gain of $\frac{2-\sqrt2\log(3-2\sqrt2)}4\approx1.399$ bits.
This shows that there are situations where LOCC outperforms simultaneous measurement.

Meanwhile if Alice and Bob try to maximize their information gain, their best strategy is for both to measure in the $Z$ basis. Alice will learn her qubit and Bob will learn his if Alice measures $\ket0$ and will learn nothing if Alice measures $\ket1$. This results in an information gain of $\frac32$ bits and a $\frac34$ probability of them guessing their state.

\section{Measurement with Entanglement}

We now turn our attention to the same setup where Alice and Bob make their measurements simultaneously, but this time equip them with a Bell pair.
In this case, Alice and Bob can also determine the state they hold with certainty.
Alice starts by measuring her qubit of the joint state in the $Z$ basis.
Depending on whether the result is $\ket0$ or $\ket1$, she measures her half of the Bell pair in the $Z$ or $X$ basis, respectively.
Meanwhile, Bob measures the two qubits he is holding in the Bell basis.

To examine this procedure in more detail, suppose Alice is handed the state $\ket0$, which she determines and then measures her entangled qubit in the $Z$ basis.
By measuring in the Bell basis, Bob learns whether his two qubits are the same or opposite.
After communicating with Alice and learning what result Alice got when she measured the entangled qubit, the two can use their results to determine the original value of his joint state qubit.
Similarly, if Alice had obtained the result $\ket1$, the same procedure works since Alice's half of the Bell pair is still measured in the same basis as Bob's qubit.

Thus the above protocol allows Alice and Bob to guarantee that they learn their state, assuming that they shared a Bell pair.
Having established some separation between both LOCC and simultaneous entangled measurements compared with simultaneous unentangled measurements, we wish to compare the two and see if one can outperform the other.

\section{LOCC}

Consider the scenario where the state is selected uniformly at random from the following six options.
Here, $\ket{+i}$ and $\ket{-i}$ represent the $+1$ and $-1$ eigenstates of the Pauli $Y$ operator.

\begin{align*}
&\ket{0}_A\ket{1}_B\\
&\ket{1}_A\ket{0}_B\\
&\ket{+}_A\ket{-}_B\\
&\ket{-}_A\ket{+}_B\\
&\ket{+i}_A\ket{-i}_B\\
&\ket{-i}_A\ket{+i}_B
\end{align*}

Note that these states can be naturally grouped into three pairs corresponding to the $X$, $Y$, and $Z$ Pauli matrices.
Now, Alice and Bob can attempt to optimize several different metrics including the probability of guessing their state, the probability of guessing the basis of their state, or the information that they can obtain about their state.
We will start by examining the case where Alice makes a single projective measurement and sends the result to Bob, who also makes a single projective measurement.
In the strategies we examine, we let Alice make a measurement and communicate classically to Bob.
This falls short of the full power of LOCC, but allows us to consider what can be done with one-way LOCC.

The best Alice and Bob can do to improve their odds of guessing their state and of guessing their basis is to have Alice measure in the $\frac{X+Z}{\sqrt2}$ basis and Bob measure in the $\frac{X-Z}{\sqrt2}$ basis.
With this protocol, Alice and Bob have a $\frac{3+2\sqrt2}{12}\approx0.486$ chance of guessing their state, a $\frac12$ chance of guessing their basis, and an information gain of \begin{align*}&\frac{(3+2\sqrt2)\log(3+2\sqrt2)}{12}\\
+&\frac{(3-2\sqrt2)\log(3-2\sqrt2)}{12}-\frac23\approx0.532.
\end{align*}

If Alice and Bob want to maximize their information gain instead, what they should do is have Alice measure in the $Z$ basis and Bob measure in the $X$ basis.
This leaves Alice and Bob with a $\frac13$ chance of guessing their state, a $\frac13$ chance of guessing their basis, and an information gain of $\frac23$ bits.

The fascinating thing is that neither of these protocols requires use of LOCC.
Alice and Bob can measure simultaneously and still achieve this result!
Naturally, since this is their optimal strategy with one-way LOCC, it must also be their optimal strategy for simultaneous unentangled measurements.

\section{Entanglement}

When Alice and Bob are allowed to share a Bell pair, we find no strategy allowing them to improve their odds of guessing their state beyond the $\approx0.486$ given by the above protocol.
However, when Alice and Bob both measure the qubit they were handed as well as their half of the Bell pair in the Bell basis, they improve both of their other metrics.
Namely, they now have a $\frac13$ chance of guessing their state, a $\frac23$ chance of guessing their basis, and an information gain of $\frac{\log(3)}2\approx0.792$ bits.

\section{Product States}

With the example in the previous two sections, it would seem that access to a Bell pair is preferable to LOCC.
However, we can construct a case where the reverse is true.
We simply suppose that Alice and Bob are either handed the state $\ket{\psi_1}_A\ket{\psi_2}_B$ or the state $\ket{\phi_1}_A\ket{\phi_2}_B$.
Notably, these are product states, which means that there is no joint measurement that cannot be expressed as a product of individual measurements on the two qubits.
Indeed, we find that using LOCC, Alice and Bob have a strategy of distinguishing between the states with the same probability as the optimal joint measurement. This fact was originally discovered in \cite{vspm}.

\begin{theorem}\rm {(Viermani et al.)} \it
LOCC can distinguish between two single-qubit product states with optimal probability.
\end{theorem}
\begin{proof}
Let us suppose that Alice and Bob are either handed the state $\ket{\psi_1}_A\ket{\psi_2}_B$ or $\ket{\phi_1}_A\ket{\phi_2}_B$ with equal probability.
Alice and Bob both start by applying a unitary transformation to their qubit that sends the respective $\ket{\psi_k}$ to $\ket0$ and $\ket{\phi_k}$ to the state $\cos(\theta_k)\ket0+\sin(\theta_k)\ket1$ where $\theta_k$ is the angle between $\ket{\psi_k}$ and $\ket{\phi_k}$.

An optimal joint measurement to distinguish between the two states $\ket\alpha$ and $\ket\beta$ separated by an angle $\theta$ is given by a projective measurement rotated from the $Z$ basis about the $Y$ axis by $\theta-\frac\pi2$.
Note that is quantity is negative.
This gives a success probability of $$\cos\left(\frac{\sin^{\text-1}(\cos(\theta))}2\right)^2=\frac{1+\sqrt{1-\cos(\theta)^2}}2.$$

In the case of Alice and Bob, we replace $\cos(\theta)$ with the inner product of the two product states, yielding a final expression of $\frac{1+\sqrt{1-\cos(\theta_1)^2\cos(\theta_2)^2}}2$.
To achieve this probability of success using LOCC, Alice starts by making an optimal measurement on her own qubit by rotating the $Z$ basis about the $Y$ axis by $\theta-\frac\pi2$.
After hearing Alice's result, Bob measures in a basis rotated by $\eta=\tan^{\text-1}\left(\frac{\sin(\theta_1)-1}{\sin(\theta_1)\cot(\theta_2)+\tan(\theta_2)}\right)$ if Alice gets the result $+1$ and $2\theta_2-\eta$ if she gets $-1$.
These angles also result in a probability of success of $\frac{1+\sqrt{1-\cos(\theta_1)^2\cos(\theta_2)^2}}2$.
\end{proof}

The idea behind the above measurement is that Alice makes a guess as to which of the states she has and then Bob either confirms or rejects this guess with his measurements.
This shows that for all values of $\theta_{1,2}$ other than when one is a multiple of $\frac\pi2$, LOCC is more powerful than a Bell pair.

\section{The Peres-Wootters State}

In a seminal paper on optimal strategies for spatially-separated observers, Peres and Wootters consider an interesting set of states \cite{pw}.
Alice and Bob receive a state chosen uniformly at random from the following list.

\begin{align*}
\ket{0}_A&\ket{0}_B\\
{\textstyle \left(\frac12\ket0+\frac{\sqrt3}2\ket{1}\right)_A}&{\textstyle\left(\frac12\ket0+\frac{\sqrt3}2\ket{1}\right)_B}\\
{\textstyle \left(\frac12\ket0-\frac{\sqrt3}2\ket{1}\right)_A}&{\textstyle\left(\frac12\ket0-\frac{\sqrt3}2\ket{1}\right)_B}
\end{align*}

In other words, Alice and Bob receive identical unentangled qubits which have been rotated from the state $\ket0$ in the $XZ$ plane by either 0, $\frac{2\pi}3$, or $\frac{4\pi}3$.
If we require Alice and Bob to each make a single measurement (disallowing the complicated algorithm of \cite{pw}), the optimal one-way LOCC strategy to maximize both their probability of guessing the state as well as their information gain is as follows.

Alice starts by making a POVM with three components each projecting along a vector directly opposite to one of the possibilities of her qubit.
For any measurement result she gets, she can guarantee that her qubit was not in the state opposite the measurement result, as the two vectors have inner product 0.
This means that Bob is left with two equally likely possibilities and they are $\frac{2\pi}3$ apart, so the optimal strategy for Bob is to use a projective measurement to distinguish the vectors.
Bob succeeds with probability $\cos(\frac\pi{12})^2=\frac{2+\sqrt3}4\approx0.933$, yielding an information gain of
$$\log(3) + \frac{\sqrt3\log\left(2 + \sqrt3\right)}2 - 2\approx1.23.$$

When Alice and Bob are allowed to share a Bell pair, their best strategy is similar. It is worth noting that although we describe these events as though they are happening in sequence to simplify understanding, Alice and Bob are making measurements simultaneously.
Alice starts by making the same POVM as in the above protocol.
Next, Bob measures his qubit and his half of the Bell pair in the Bell basis.
This has effectively teleported Bob's state to Alice half of the Bell pair but with a random Pauli transformation applied to it.
Then, Alice measures her half of the Bell pair.
If Alice determined that the state is not $\ket{00}$, she makes a measurement in the $X$ basis, otherwise she measures in the $Z$ basis.
The $X$ basis measurement is optimal for distinguishing between the two remaining states.
However, the $Z$ basis measurement is only optimal because Alice does not know which Pauli transformation was applied to Bob's state during the quantum teleportation.
The result of this protocol is that Alice and Bob have a $\frac{9+\sqrt3}{12}\approx0.894$ chance of guessing their state and get an information gain of \fontsize{11.5pt}{11.5}\selectfont
$$\log(3)+\frac{2\sqrt3\log(2+\sqrt3)-5\log(5)}{12}\approx1.17.$$

\fontsize{12pt}{12pt}\selectfont
Additionally, to see how much of an improvement this is, we can compare these values to the best strategy for Alice and Bob when they are forced to make their measurements simultaneously but without the use of entanglement.
This does not allow them to benefit from neither LOCC nor a Bell pair.
We find that their optimal strategy is to make a three component POVM, again with each component offset by $\frac{\pi}3$, but this time the vectors along which the POVM elements are projecting have been rotated from our original vectors by $\frac12\tan^{\text-1}\left(\frac{\sqrt[3]4-1}{\sqrt3}\right)\approx0.163$.
This gives Alice and Bob a $\frac1{6-3\sqrt[3]4}\approx0.808$ chance of guessing their state and an information gain of approximately $0.867$ bits.

To maximize their information gain without LOCC or a Bell pair, Alice and Bob should both measure with three POVM components each pointing opposite one of the three possible vectors that they are holding.
This will result in them getting different results and identifying their state or the same result and having to guess between two possibilities, leading to an information gain of $\log(3)-\frac12\approx1.08$ bits and a probability of guessing their state of $\frac34$.

\section{Conclusion}

In this paper we have examined various strategies for two observers to try to either determine an unknown state or maximize their learned information about it.
We compared how much these metrics improved when they were allowed to use one-way LOCC compared with when they were allowed to share a Bell pair.

We find that in every situation we examined, Alice and Bob have a better chance of guessing their state when they are allowed to use LOCC compared with the Bell pair. It is an interesting question whether this holds for all sets of product states, 

In particular, when distinguishing only two states, LOCC allowed Alice and Bob to identify their state with optimal probability. 
However, in one scenario with six possible choices for the state, the Bell pair was more helpful than LOCC when it came to identifying the basis of the state as well as maximizing the information gain.
This might allow Alice and Bob to amplify this advantage by sending a long sequence of such states with a Bell pair for each one.
However, this opens the door for Alice and Bob to make joint measurements on all of their qubits, which may show that LOCC dominates once again for maximizing the probability of guessing their state correctly.
We believe this to be an interesting direction for further research, so we ask: does LOCC always allow for a higher probability of guessing the state than simultaneous measurement with entanglement, despite the fact that the latter sometimes has higher information gain?

\section{Acknowledgements}

The authors are grateful to the NSF for supporting this research.
The authors are also grateful to William Wootters for helpful discussions.

Andrey Boris Khesin was supported in part by NSF grant CCF-1452616.

Peter Shor was supported in part by the National Science Foundation under Grant No. CCF-1729369 and through the NSF Science and Technology Center for Science of Information under Grant No. CCF-0939370.

\end{document}